\documentclass[envcountsect]{llncs}

\usepackage{amsmath}
\usepackage{amssymb}
\usepackage{comment}
\usepackage{times}
\usepackage{wrapfig}
\usepackage{float}
\usepackage{gastex}
\usepackage{graphicx}

\newcommand{\li}{\liminf_{n \rightarrow \infty}}
\newcommand{\limn}{\lim_{n \rightarrow \infty}}
\newcommand{\si}{\sum_{i=0}^\infty}
\newcommand{\ls}{\limsup_{n \rightarrow \infty}}

\newcommand{\sign}{{\sf sign}}
\newcommand{\abs}[1]{\ensuremath{\lvert #1\rvert}}
\newcommand{\tuple}[1]{\langle #1 \rangle}
\newcommand{\weight}{w}

\newcommand{\nat}{\mathbb N}

\newcommand{\rat}{{\mathbb Q}}

\newcommand{\real}{{\mathbb R}}

\newcommand{\pat}{\omega}
\newcommand{\Pat}{\Omega}
\newcommand{\straa}{\sigma}
\newcommand{\Straa}{\Sigma}
\newcommand{\strab}{\pi}
\newcommand{\Strab}{\Pi}

\newcommand{\val}{val}

\newcommand{\MeanPayoff}{{\sf MeanPayoff}}
\newcommand{\DiscSum}{{\sf DiscSum}}
\newcommand{\WAvg}{{\sf WeightedAvg}}

\begin{document}

\pagestyle{plain}

\title{On Memoryless Quantitative Objectives}

\author{Krishnendu Chatterjee$^1$ \and Laurent Doyen$^2$ \and 
  Rohit Singh$^3$}

\institute{ Institute of Science and Technology(IST) Austria \and LSV, ENS Cachan \& CNRS, France
\and Indian Institute of Technology(IIT) Bombay
}

\maketitle

\begin{abstract}
In two-player games on graph, the players construct an infinite path through
the game graph and get a reward computed by a payoff function over infinite paths.
Over weighted graphs, the typical and most studied payoff functions compute the limit-average 
or the discounted sum of the rewards along the path. 
Beside their simple definition, these two payoff functions enjoy the property that memoryless
optimal strategies always exist.

In an attempt to construct other simple payoff functions, we define a class 
of payoff functions which compute an (infinite) weighted average of the rewards.
This new class contains both the limit-average and discounted sum functions,
and we show that they are the only members of this class which 
induce memoryless optimal strategies, showing that there is essentially no other 
simple payoff functions.
\end{abstract}

\section{Introduction}

Two-player games on graphs have many applications in computer science, 
such as the synthesis problem~\cite{Church62}, and the model-checking of
open reactive systems~\cite{AHK02}. Games are also fundamental in logics,
topology, and automata theory~\cite{Kechris,GradelThoWil02,Martin75}. 
Games with quantitative objectives have been used to design resource-constrained
systems~\cite{ZP96,dAHM03,CAHS03,CD10a}, and to support quantitative model-checking
and robustness~\cite{CDH10b,CHJS10,RV11}.

In a two-player game on a graph, a token is moved by the players along
the edges of the graph. The set of states is partitioned into 
player-1 states from which player~$1$ moves the token, 
and player-2 states from which player~$2$ moves the token.
The interaction of the two players results in a play, an infinite path 
through the game graph. In qualitative zero-sum games, each play
is winning for one of the player; in quantitative games, a payoff function 
assigns a value to every play, which is paid by player~$2$ to 
player~$1$. Therefore, player~$1$ tries to maximize the payoff while 
player~$2$ tries to minimize it. Typically, the edges of the graph carry a reward,
and the payoff is computed as a function of the infinite sequences of rewards
on the play.

Two payoff functions have received most of the attention in literature:
the \emph{mean-payoff} function (for example, see~\cite{EM79,ZP96,GKK88,LigLip69,FV97,MN81})
and the \emph{discounted-sum} function (for example, see~\cite{Sha53,FV97,Puri95,Puterman,dAHM03}).
The mean-payoff value is the long-run average of the rewards. The discounted sum
is the infinite sum of the rewards under a discount factor $0 < \lambda < 1$. 
For an infinite sequence of rewards $w = w_0 w_1 \dots$, we have:
$$ \MeanPayoff(w) = \liminf_{n \to \infty} \frac{1}{n} \cdot \sum_{i=0}^{n-1} w_i \qquad 
\DiscSum_\lambda(w) = (1-\lambda) \cdot \sum_{i=0}^{\infty} \lambda^i \cdot  w_i$$
While these payoff functions have a simple, intuitive, and mathematically elegant
definition, it is natural to ask why they are playing such a central role in 
the study of quantitative games. One answer is perhaps that \emph{memoryless} optimal
strategies exist for these objectives. A strategy is memoryless if it is 
independent of the history of the play and depends only on the current state.
Related to this property is the fact that the problem of deciding the
winner in such games is in NP $\cap$ coNP, while no polynomial time algorithm
is known for this problem. The situation is similar to the case of parity games
in the setting of qualitative games where it was proved that the parity objective
is the only prefix-independent objective to admit memoryless winning strategies~\cite{CN06},
and the parity condition is known as a canonical way to express $\omega$-regular
languages~\cite{Thomas97}.

In this paper, we prove a similar result in the setting of quantitative games.
We consider a general class of payoff functions which compute an infinite 
weighted average of the rewards. The payoff functions are parameterized by an
infinite sequence of rational coefficients $\{c_n\}_{n\geq 0}$, and defined as follows:
\[
\WAvg(w) = \liminf_{n \to \infty} \frac{\sum_{i=0}^n c_i \cdot w_i}{\sum_{i=0}^n c_i}.
\]
We consider this class of functions for its simple and natural definition, and because
it generalizes both mean-payoff and discounted-sum which can be obtained as special cases, 
namely for $c_i = 1$ for all\footnote{Note that other sequences also define the mean-payoff function, such as $c_i = 1 + 1 / 2^i$.}
 $i \geq 0$, and $c_i = \lambda^i$ respectively.
We study the problem of characterizing which payoff functions in this
class admit memoryless optimal strategies for both players. 
Our results are as follows:
\begin{enumerate}
\item If the series $\sum_{i=0}^{\infty} c_i$ converges (and is finite), then discounted sum 
is the \emph{only} payoff function that admits memoryless optimal strategies for both 
players.

\item If the series $\sum_{i=0}^{\infty} c_i$ does not converge, but the sequence 
$\{c_n\}_{n\geq 0}$ is bounded, then for memoryless optimal strategies the payoff 
function is equivalent to the mean-payoff function (equivalent for the optimal 
value and optimal strategies of both players). 
\end{enumerate}

Thus our results show that the discounted sum and mean-payoff functions, beside their
elegant and intuitive definition, are the only members from a large 
class of natural payoff functions that are simple (both players have memoryless
optimal strategies).
In other words, there is essentially no other simple payoff functions in the class
of weighted infinite average payoff functions. This further establishes 
the canonicity of the mean-payoff and discounted-sum functions, and suggests that they
should play a central role in the emerging theory of quantitative automata and 
languages~\cite{DrosteG07,Henzinger10,Bojanczyk10,CDH10b}.

In the study of games on graphs, characterizing the classes of payoff functions 
that admit memoryless strategies is a research direction that has been investigated
in the works of~\cite{GimZie04,GimZie05} which give general conditions on the payoff 
functions such that both players have memoryless optimal strategies, 
and~\cite{Kop06} which presents similar results when only one player has memoryless
optimal strategies.
The conditions given in these previous works are useful in this paper, in particular
the fact that it is sufficient to check that memoryless strategies are sufficient
in one-player games~\cite{GimZie05}. However, conditions such as sub-mixing and 
selectiveness of the payoff function are not immediate to establish, especially 
when the sum of the coefficients $\{c_n\}_{n\geq 0}$ does not converge. We identify
the necessary condition of boundedness of the coefficients $\{c_n\}_{n\geq 0}$ 
to derive the mean-payoff function.
Our results show that if the sequence is convergent, then discounted sum 
(specified as $\{\lambda^n\}_{n\geq 0}$, for $\lambda<1$) is the only 
memoryless payoff function; and if the sequence is divergent and bounded, then 
mean-payoff (specified as $\{\lambda^n\}_{n\geq 0}$ with $\lambda=1$)
is the only memoryless payoff function.
However we show that if the sequence is divergent and unbounded, then there 
exists a sequence $\{\lambda^n\}_{n\geq 0}$, with $\lambda>1$, that does not
induce memoryless optimal strategies.



\section{Definitions}

\noindent{\bf Game graphs.}
A two-player \emph{game graph} $G=\tuple{Q, E, \weight}$ consists of a finite set $Q$ of states 
partitioned into \mbox{player-$1$} states $Q_1$ and player-2 states $Q_2$ (i.e., $Q=Q_1 \cup Q_2$),
and a set $E \subseteq Q \times Q$ of edges such that for all $q \in Q$,
there exists (at least one) $q' \in Q$ such that $(q,q') \in E$. 
The weight function $\weight:E \to \rat$ assigns a reward to each edge.
For a state $q \in Q$, we write $E(q)=\{r \in Q \mid (q,r) \in E \}$ 
for the set of successor states of~$q$.
A \emph{player-$1$ game} is a game graph where $Q_1 = Q$ and $Q_2 = \emptyset$.
Player-$2$ games are defined analogously.


\smallskip\noindent{\bf Plays and strategies.}
A game on $G$ starting from a state $q_0 \in Q$ is played in rounds as follows. 
If the game is in a player-1 state, then player~$1$ chooses the successor state from the set
of outgoing edges; otherwise the game is in a player-$2$ state, and player $2$ chooses the successor 
state. 
The game results in a \emph{play} from~$q_0$, i.e., 
an infinite path $\rho = \tuple{q_0 q_1 \dots}$ such that $(q_i,q_{i+1}) \in E$ for all $i \geq 0$. 
We write $\Pat$ for the set of all plays.
The prefix of length $n$ of $\rho$ is denoted by $\rho(n) = q_0 \dots q_n$. 
A strategy for a player is a recipe that specifies how to extend plays.
Formally, a \emph{strategy} for player~$1$ is a function
$\straa: Q^* Q_1 \to Q$ such that $(q,\straa(\rho\cdot q)) \in E$ for all $\rho \in Q^*$ 
and $q \in Q_1$. The strategies for player~2 are defined analogously.
We write $\Straa$ and $\Strab$ for the sets of all strategies for 
player~1 and player~2, respectively.

An important special class of strategies are \emph{memoryless} strategies
which do not depend on the history of a play, but only on the current state. 
Each memoryless strategy for player~1 can be specified as a function
$\straa$: $Q_1 \to Q$ such that $\straa(q) \in E(q)$ for all $q \in Q_1$,
and analogously for memoryless player~2 strategies.

Given a starting state $q \in Q$, the \emph{outcome} of strategies $\straa \in \Straa$ for player~1, 
and $\strab \in \Strab$ for player~2, is the play 
$\pat(s,\straa,\strab) = \tuple{q_0 q_1 \dots} $ such that : 
$q_0=q$ and for all $k \geq 0$,
if $q_k \in Q_1$, then $\straa(q_0,q_1,\ldots,q_k)=q_{k+1}$, and
if $q_k \in Q_2$, then $\strab(q_0,q_1,\ldots,q_k)=q_{k+1}$.

\smallskip\noindent{\bf Payoff functions, optimal strategies.}
The objective of player~$1$ is to construct a play that maximizes a \emph{payoff function}
$\phi: \Pat \to \real \cup \{-\infty, +\infty\}$ which is a measurable function 
that assigns to every value a real-valued payoff. 
The value for player~$1$ is the maximal payoff that can be achieved 
against all strategies of the other player.
Formally the value for player~1 for a starting state $q$ is defined as 
\[
\val_1(\phi) =  \sup_{\straa \in \Straa} \inf_{\strab \in \Strab} \phi(\pat(q,\straa,\strab)).
\]
A strategy $\straa^*$ is \emph{optimal} for player~1 from $q$ if the strategy 
achieves at least the value of the game against all strategies for player~2, 
i.e.,
\[
\inf_{\strab \in \Strab} \phi(\pat(q,\straa^*,\strab)) = \val_1(\phi).
\]
The values and optimal strategies for player~2 are defined analogously.

%


The mean-payoff and discounted-sum functions are examples of payoff functions
that are well studied, probably because they are simple in the sense that 
they induce memoryless optimal strategies and that this property yields conceptually
simple fixpoint algorithms for game solving~\cite{Sha53,EM79,ZP96,FV97}. 
In an attempt to construct other simple payoff functions, we define the class 
of \emph{weighted average payoffs} which compute (infinite) weighted averages of the rewards,
and we ask which payoff functions in this class induce memoryless optimal strategies.

We say that a sequence $\{c_n\}_{n \geq 0}$ of rational
numbers has \emph{no zero partial sum} if $\sum_{i=0}^n c_i \neq 0$ for all $n \geq 0$. 
Given a sequence $\{c_n\}_{n \geq 0}$ with no zero partial sum, the \emph{weighted average payoff function} 
for a play $\tuple{q_0 q_1 q_2 \ldots}$ is
\[
\phi\left(q_0 q_1 q_2  \dots  \right) = \li \frac{ \sum_{i=0}^n c_i \cdot w(q_i,q_{i+1}) }{\sum_{i=0}^n c_i}.
\]

Note that we use $\li$ in this definition because the plain limit may not exist
in general. The behavior of the weighted average payoff functions
crucially depends on whether the series $S = \sum_{i=0}^{\infty} c_i$ converges or not.
In particular, the plain limit exists if $S$ converges (and is finite). Accordingly,
we consider the cases of converging and diverging sum of weights to characterize  
the class of weighted average payoff functions that admit memoryless optimal strategies for both players.
Note that the case where $c_i = 1$ for all $i \geq 0$ gives the mean-payoff
function (and $S$ diverges), and the case $c_i = \lambda^i$ for $0 < \lambda < 1$ gives the discounted sum
with discount factor $\lambda$ (and $S$ converges). 
All our results hold if we consider $\ls$ instead of $\li$ in the definition of 
weighted average objectives.

In the sequel, we consider payoff functions $\phi: \rat^\omega \to \real$ 
with the implicit assumption that the value of a play $q_0 q_1 q_2 \dots \in Q^\omega$
according to $\phi$ is $\phi(\weight(q_0,q_1) \weight(q_1,q_2) \dots)$
since the sequence of rewards determines the payoff value.


We recall the following useful necessary condition for memoryless
optimal strategies to exist~\cite{GimZie05}. A payoff function $\phi$
is \emph{monotone} if whenever there exists a finite sequence of
rewards $x \in \rat^*$ and two sequences $u,v \in \rat^\omega$ 
such that $\phi(xu) \leq \phi(xv)$, then $\phi(yu) \leq \phi(yv)$
for all finite sequence of rewards $y \in \rat^*$.

\begin{lemma}[\cite{GimZie05}] \label{prefc}
If the payoff function $\phi$ induces memoryless optimal strategy for all two-player game graphs, 
then $\phi$ is monotone.
\end{lemma}

\section{Weighted Average with Converging Sum of Weights}\label{sec:converging}


%
%
The main result of this section is that for converging sum of weights 
(i.e., if $\lim_{n \to \infty} \sum_{i=0}^{n} c_i = c^* \in \real$),
the only weighted average payoff function that induce memoryless optimal strategies
is the discounted sum.

\begin{theorem}\label{thm:disc}
Let $(c_n)_{n \in \nat}$ be a sequence of real numbers with no zero partial sum such that $\sum_{i=0}^\infty c_i = c^* \in \real$. 
The weighted average payoff function defined by $(c_n)_{n \in \nat}$ induces optimal memoryless 
strategies for all $2$-player game graphs if and only if
there exists $0 \leq \lambda < 1$ such that $c_{i+1} = \lambda \cdot c_i$ for all $i \geq 0$.
\end{theorem}

To prove Theorem~\ref{thm:disc}, we first use its assumptions to obtain necessary conditions for
the weighted average payoff function defined by $(c_n)_{n \in \nat}$ to 
induce optimal memoryless strategies. By \emph{assumptions of Theorem~\ref{thm:disc}},
we refer to the fact that $(c_n)_{n \in \nat}$ is a sequence of real numbers with 
no zero partial sum such that $\sum_{i=0}^\infty c_i = c^* \in \real$, and that it defines 
a weighted average payoff function that induces optimal memoryless 
strategies for all $2$-player game graphs. All lemmas of this section use the
the assumptions of Theorem~\ref{thm:disc}, but we generally omit to mention them.

Let $d_n = \sum_{i=0}^{n-1} c_i$, $l = \li \frac{1}{d_n}$ and $L = \ls \frac{1}{d_n}$.
The assumption that $\sum_{i=0}^\infty c_i = c^* \in \real$ implies that $l \neq 0$.

Note that $c_0 \neq 0$ since $(c_n)_{n \in \nat}$ is a sequence with no zero partial sum. 
We can define the sequence $c_n' = \frac{c_n}{c_0}$ which defines the same payoff 
function $\phi$. Therefore we assume without loss of generality that $c_0 = 1$.

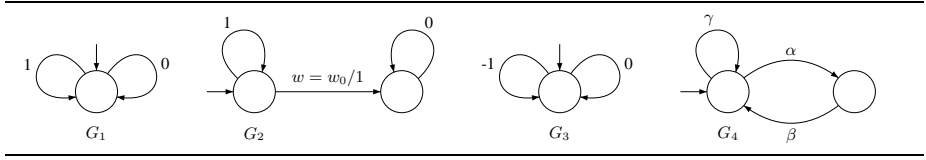
\begin{figure}[t]
\hrule
\begin{center}

\scalebox{0.7}{\begin{tabular}{c c c c}

\begin{picture}(32,25)(0,0)
\node[Nmarks=i, iangle=90](n0)(15,10){}
\drawloop[loopCW=n, ELside=r, loopangle=158](n0){1}
\drawloop[loopangle=22](n0){0}
\node[Nframe=n,NLangle=0,Nmr=0](n2)(15,2){$G_1$}
\end{picture}

& 

\begin{picture}(52,25)(0,0)

\node[Nmarks=i](n0)(12,10){}
\node(n1)(40,10){}

\drawloop[ELpos=45, loopangle=105](n0){1}
\drawloop[ELpos=45, loopCW=n, ELside=r, loopangle=75](n1){0}

\node[Nframe=n](n2)(12,2){$G_2$}
\drawedge(n0,n1){$w = w_0/1$}

\end{picture}

& 

\begin{picture}(36,25)(0,0)
\node[Nmarks=i, iangle=90](n0)(17,10){}
\drawloop[loopCW=n, ELside=r, loopangle=158](n0){-1}
\drawloop[loopangle=22](n0){0}
\node[Nframe=n,NLangle=0,Nmr=0](n2)(17,2){$G_3$}
\end{picture}

&

\begin{picture}(46,27)(0,0)

\node[Nmarks=i](n3)(12,10){}
\node(n4)(36,10){}

\drawloop[loopangle=105](n3){$\gamma$}
\drawedge[curvedepth=6](n3,n4){$\alpha$}
\drawedge[curvedepth=6](n4,n3){$\beta$}
\node[Nframe=n](n2)(12,2){$G_4$}
\end{picture}
\end{tabular}
}
\end{center}
\vspace{-3mm}
\hrule
\caption{Examples of one-player game graphs.} 
\label{fig:g1}
\end{figure}

\begin{lemma}\label{lem:0lL1}
If the weighted average payoff function defined by $(c_n)_{n \in \nat}$ induces optimal memoryless 
strategies for all $2$-player game graphs, then $0 \leq l \leq L \leq 1$.
\end{lemma}

\begin{proof}
Consider the one-player game graph $G_1$ shown in \figurename~\ref{fig:g1}. 
In one-player games, strategies correspond to paths. 
The two memoryless strategies give the paths $0^\omega$ and $1^\omega$ with payoff 
value~$0$ and~$1$ respectively. 
The strategy which takes edge with reward $1$ once, and then always the edge
with reward $0$ gets payoff $\phi\left(10^\omega\right) = \li \frac{1}{d_n} = l$. 
Similarly, the path $01^\omega$ has payoff $\phi \left(01^\omega\right)$ = 
$\li \left(1-\frac{1}{d_n} \right) = 1 - \ls \frac{1}{d_n} = 1 - L$. 
Since all such payoffs must be 
between the payoffs obtained by the only two memoryless strategies, we have
$l \geq 0$ and $L \leq 1$, and the result follows ($L \geq l$ follows
from their definition).
\qed
\end{proof}

%
%
%
%

\begin{lemma}\label{lem:ckdk}
There exists $w_0 \in \nat$ such that $w_0 l > 1$
and the following inequalities hold, for all $k \geq 0$:
$c_k l \leq  1 - d_k L$ and $c_k w_0 l \geq 1 - d_k L$.
\end{lemma}

\begin{proof}
Since $1 \geq l > 0$ (by Lemma~\ref{lem:0lL1}), we can choose $w_0 \in \nat$ such that $w_0 l > 1$.
Consider the game graph $G_2$ shown in \figurename~\ref{fig:g1} and the case when $w = 1$. 
The optimal memoryless strategy is to stay on the starting state forever because $\phi(10^\omega) = l \leq \phi(1^w) = 1$.
Using Lemma \ref{prefc}, we conclude that since $\phi(1 0^\omega) \leq \phi(1^\omega)$, we must have 
$\phi(0^k 1 0^\omega) \leq \phi(0^k 1^\omega)$ i.e. $c_k l \leq 1 - \left(\sum_{i=0}^{k-1} c_i\right)L$ 
which implies $c_k l \leq 1 - d_k L$.

Now consider the case when $w = w_0$ in \figurename~\ref{fig:g1}. The optimal memoryless strategy is to choose 
the edge with reward $w_0$ from the starting state since $\phi(w_0 0^\omega) = w_0l > \phi(1^\omega) = 1$.
Using Lemma \ref{prefc}, we conclude that since $\phi(w_0 0^\omega) > \phi(1^\omega)$, we must have 
$\phi(0^k w_0 0^\omega) \geq \phi(0^k 1^\omega)$ i.e. $c_k w_0 l \geq 1 - \left(\sum_{i=0}^{k-1} c_i\right)L$
which implies $c_k w_0 l \geq 1 - d_k L$.
\qed
\end{proof}

From the inequalities in Lemma~\ref{lem:ckdk}, it is easy to see that since $w_0 > 1$ we must have $c_k \geq 0$ for all $k$.

\begin{corollary}\label{coro:positive}
Assuming $c_0 = 1$, we have $c_k \geq 0$ for all $k\geq 0$.
\end{corollary}

It follows from Corollary~\ref{coro:positive} that the sequence $(d_n)_{n \geq 0}$ 
is increasing and bounded from above (if $d_n$ was not 
bounded, then there would exist a subsequence $(d_{n_k})$ which diverges, implying that the sequence 
$\{\frac{1}{d_{n_k}}\}$ converges to $0$ in contradiction with the fact that $\li \frac{1}{d_n} = l > 0$). 
Therefore, $d_n$ must converge to some real number say $c^* > 0$ (since $c_0 = 1$). 
We need a last lemma to prove Theorem~\ref{thm:disc}. 
Recall that we have $c_i \geq 0$ for all $i$ and $\sum_{i=0}^\infty c_i = c^* > 0$. 
Given a finite game graph~$G$, let $W$ be the largest reward in absolute value.
For any sequence of rewards $(w_n)$ in a run on $G$, the sequence
$\chi_n = \sum_{i=0}^n c_i (w_i + W)$ is increasing and bounded from above by $2 \cdot W d_n$ and thus 
by $2 \cdot W c^*$. 
Therefore, $\chi_n$ is a convergent sequence and  $\sum_{i=0}^\infty c_i w_i$ converges as well. 
Now, we can write the payoff function as $\phi(w_0w_1 \dots ) = \frac{\sum_{i=0}^\infty c_i w_i}{c^*}$. 
We decompose $c^*$ into $S_0 = \si c_{2i}$ and $S_1 = \si c_{2i+1}$, i.e. $c^* = S_0 + S_1$. 
Note that $S_0$ and $S_1$ are well defined.

\begin{lemma} \label{lem:s0s1}
If there exist numbers $\alpha, \beta, \gamma$ such that $\alpha S_0 + \beta S_1 \leq \gamma (S_0 + S_1)$, 
then $(\gamma - \alpha) c_i \geq (\beta - \gamma) c_{i+1}$ for all $i \geq 0$.
\end{lemma}

\begin{proof}
Consider the game graph $G_4$ as shown in \figurename~\ref{fig:g1}. 
The condition $\alpha S_0 + \beta S_1 \leq \gamma (S_0 + S_1)$ 
implies that the optimal memoryless strategy is to always choose the edge with reward~$\gamma$. 
This means that $\phi(\gamma^i \alpha \beta \gamma^\omega) \leq \phi(\gamma^\omega)$ hence 
$\alpha c_i + \beta c_{i+1} \leq \gamma (c_i + c_{i+1})$, i.e. $(\gamma - \alpha) c_i \geq (\beta-\gamma) c_{i+1}$ for all $i \geq 0$.
\qed
\end{proof}

We are now ready to prove the main theorem of this section. 

\begin{proof}[of Theorem~\ref{thm:disc}]
First, we show that $S_1 \leq S_0$. By contradiction, assume that $S_1 > S_0$.
Choosing $\alpha = 1$, $\beta = -1$, and $\gamma = 0$ in Lemma~\ref{lem:s0s1}, and since $S_0 - S_1 \leq 0$, 
we get $-c_i \geq -c_{i+1}$ for all $i \geq 0$ which implies $c_n \geq c_0 = 1$ for all $n$,
which contradicts that $\si c_i$ converges to $c^* \in \real$.

Now, we have $S_1 \leq S_0$ and let $\lambda = \frac{S_1}{S_0} \leq 1$. 
Consider a sequence of rational numbers $\frac{l_n}{k_n}$ converging to 
$\lambda$ from the right, i.e., $\frac{l_n}{k_n} \geq \lambda$ for all $n$, and $\lim_{n \to \infty} \frac{l_n}{k_n} = \lambda$. 
Taking $\alpha = 1$, $\beta = k_n + l_n + 1$, and $\gamma = l_n + 1$
in Lemma~\ref{lem:s0s1}, and since the condition $S_0 + (k_n + l_n + 1)S_1 \leq (l_n +1)(S_0 + S_1)$
is equivalent to $k_n S_1 \leq l_n S_0$ which holds since $\frac{l_n}{k_n} \geq \lambda$, 
we obtain $l_n c_i \geq k_n c_{i+1}$ for all $n \geq 0$ and all $i \geq 0$,
that is $c_{i+1} \leq \frac{l_n}{k_n} c_i$ and in the limit for $n \to \infty$,
we get $c_{i+1} \leq \lambda c_i$ for all $i \geq 0$.

Similarly, consider a sequence of rational numbers $\frac{r_n}{s_n}$ converging to 
$\lambda$ from the left.
Taking $\alpha = r_n + s_n + 1$, $\beta = 1$, and $\gamma = s_n + 1$
in Lemma~\ref{lem:s0s1}, and since the condition $(r_n + s_n + 1) S_0 + S_1 \leq (s_n +1) (S_0+ S_1)$
is equivalent to $r_n S_0 \leq s_n S_1$ which holds since $\frac{r_n}{s_n} \leq \lambda$, 
we obtain $r_n c_i \leq s_n c_{i+1}$ for all $n \geq 0$ and all $i \geq 0$,
that is $c_{i+1} \geq \frac{r_n}{s_n} c_i$ and in the limit for $n \to \infty$,
we get $c_{i+1} \geq \lambda c_i$ for all $i \geq 0$.

The two results imply that $c_{i+1} = \lambda c_i$ for all $i \geq 0$ where $0 \leq \lambda < 1$. 
Note that $\lambda \neq 1$ because $\si c_i$ converges. 
\qed
\end{proof}

Since it is known that for $c_i = \lambda^i$, the weighted average payoff function
induces memoryless optimal strategies in all two-player games, Theorem \ref{thm:disc} 
shows that discounted sum is the only memoryless payoff function when the sum
of weights $\si c_{i}$ converges.

\section{Weighted Average with Diverging Sum of Weights}
In this section we consider weighted average objectives such that 
the sum of the weights $\si c_{i}$ is divergent. 
We first consider the case when the sequence $(c_n)_{n \in \nat}$ is bounded and show 
that the mean-payoff function is the only memoryless one.



\subsection{Bounded sequence}\label{sec:bounded}
We are interested in characterizing the class of weighted average objectives 
that are memoryless, under the assumption the sequence $(c_n)$
is \emph{bounded}, i.e., there exists a constant $c$ such that $\abs{c_n} \leq c$
for all $n$. The boundedness assumption is satisfied by the important special case of regular 
sequence of weights which can be produced by a deterministic finite automaton.
We say that a sequence $\{c_n\}$ is \emph{regular} if it is eventually periodic,
i.e. there exist $n_0 \geq 0$ and $p > 0$ such that $c_{n+p} = c_n$ for all $n \geq n_0$.
Recall that we assume the partial sum to be always non-zero,
i.e., $d_n = \sum_{i=0}^{n-1} c_i \neq 0$ for all $n$.
We show the following result.

\begin{theorem}\label{thm:mean}
Let $(c_n)_{n \in \nat}$ be a sequence of real numbers with no zero partial sum such that 
$\sum_{i=0}^\infty \abs{c_i} = \infty$ (the sum is divergent) and there exists a constant $c$ such
that $\abs{c_i} \leq c$ for all $i \geq 0$ (the sequence is bounded). 
The weighted average payoff function $\phi$ defined by $(c_n)_{n \in \nat}$ 
induces optimal memoryless strategies for all $2$-player game graphs if and 
only if $\phi$ coincides with the mean-payoff function over regular words.
\end{theorem}

\paragraph{Remark.}
From Theorem~\ref{thm:mean}, it follows that all mean-payoff functions~$\phi$ over
bounded sequences that induce optimal memoryless strategies are equivalent
to the mean-payoff function, in the sense that the optimal value and optimal strategies 
for~$\phi$ are the same as for the mean-payoff function. This is because 
memoryless strategies induce a play that is a regular word.
We also point out that it is not necessary that the sequence $(c_n)_{n \geq 0}$ 
consists of a constant value to define the mean-payoff function.
For example, the payoff function defined by the sequence $c_n = 1 + 1/(n+1)^2$ 
also defines the mean-payoff function.
\smallskip

We prove Theorem~\ref{thm:mean} through a sequence of lemmas.
In the following lemma we prove the existence of the limit of the 
sequence $\{\frac{1}{d_n}\}_{n \geq 0}$.

\begin{lemma} \label{lem:sup0}
If $\li \frac{1}{d_n} = 0$, then $\ls \frac{1}{d_n} = 0$.
\end{lemma}
%

\begin{proof}
Since $l = \li \frac{1}{d_n} = 0$, there is a subsequence $\{d_{n_k}\}$ which either diverges to $+\infty$ or $-\infty$.

1. If the subsequence $\{d_{n_k}\}$ diverges to $+\infty$, assume without loss 
of generality that each $d_{n_k} > 0$.
Consider the one-player game graph $G_3$ shown in Figure~\ref{fig:g1}. We consider the run corresponding to 
taking the edge with weight $-1$ for the first $n_k$ steps followed by taking the $0$ edge forever. 
The payoff for this run is given by 
\[
\li \frac{-d_{n_k}}{d_n} = -d_{n_k} \cdot \ls \frac{1}{d_n} = -d_{n_k} \cdot L.
\] 
Since we assume existence of memoryless optimal strategies this payoff should lie between $-1$ 
and $0$.
This implies that $d_{n_k} \cdot L \leq 1$ for all $k$.
Since $L \geq l \geq 0$ and the sequence $d_{n_k}$ is unbounded,
we must have $L=0$.

2. If the subsequence $\{d_{n_k}\}$ diverges to $-\infty$, assume that each $d_{n_k} < 0$.
Consider the one-player game graph $G_1$ shown in Figure \ref{fig:g1}. We consider the run corresponding to 
taking the edge with weight $1$ for the first $n_k$ steps followed by taking the $0$ edge forever. 
The payoff for this run is given by 
\[
\li \frac{d_{n_k}}{d_n} = -\abs{d_{n_k}}\cdot \ls \frac{1}{d_n} = -\abs{d_{n_k}} \cdot L.
\] 
This payoff should lie between $0$ and $1$ (optimal strategies being memoryless),
and this implies $L=0$ as above.
\qed
\end{proof}

Since $\ls d_n = \infty$, 
Lemma \ref{lem:sup0} concludes that the sequence $\{\frac{1}{d_n}\}$ converges to $0$ i.e. 
$\lim_{n\rightarrow \infty} \frac{1}{d_n} = 0$. It also gives us the following 
corollaries which are a simple consequence of the fact that $\li (a_n + b_n) = a + \li b_n$ 
if $a_n$ converges to $a$.

\begin{corollary}\label{cor:preind}
If $l=0$, then the payoff function $\phi$ does not depend upon any finite prefix of the run, i.e.,
$\phi(a_1a_2 \dots a_k u) = \phi(0^ku) = \phi(b_1b_2 \dots b_ku)$ for all $a_i$'s and $b_i$'s.
\end{corollary}

\begin{corollary}\label{cor:cchange}
If $l=0$, then the payoff function $\phi$ does not change by modifying finitely many values in the sequence $\{c_n\}_{n \geq 0}$. 
\end{corollary}

By Corollary \ref{cor:preind}, we have $\phi(xa^\omega) = a$ for all $a \in \real$.
For $0 \leq i \leq k-1$, consider the payoff $S_{k,i} = \phi\left( (0^i10^{k-i-1})^\omega \right)$
for the infinite repetition of the finite sequence of $k$ rewards in which all
rewards are $0$ except the $(i+1)$th which is $1$.
We show that $S_{k,i}$ is independent of $i$.

\begin{lemma}\label{lemm:sk0}
We have  $S_{k,0} = S_{k,1} =  \dots  = S_{k,k-1} \leq \frac{1}{k}$.
\end{lemma}

\begin{proof}
 If $S_{k,0} \leq S_{k,1}$ then by prefixing by the single letter word $0$ and using Lemma~\ref{prefc} 
we conclude that $S_{k,1} \leq S_{k,2}$. 
We continue this process until we get $S_{k,k-2} \leq S_{k,k-1}$. 
After applying this step again we get 
\[
S_{k,k-1} \leq \phi\left(0(0^{k-1}1)^\omega \right) = \phi\left(1(0^{k-1}1)^\omega \right) 
= \phi\left((10^{k-1})^\omega \right) = S_{k,0}.
\] 
Hence, we have $S_{k,0} \leq S_{k,1} \leq \dots  \leq S_{k,k-1} \leq S_{k,0}$.
Thus we have  $S_{k,i}$ is a constant irrespective of the value of $i$. 
A similar argument works in the other case when $S_{k,0} \geq S_{k,1}$.

Using the fact that $\li (a_{1,n} + a_{2,n} + \dots  + a_{k,n}) \geq \li a_{1,n} + \dots  + \li a_{k,n}$, we get that
$S_{k,i} \leq \frac{1}{k}$ for $0 \leq i \leq k-1$.
\qed
\end{proof}

Let $T_{k,i} = -\phi\left( (0^i(-1)0^{k-i-1})^\omega \right)$. By similar 
argument as in the proof of Lemma~\ref{lemm:sk0}, we show that 
$T_{k,0} = T_{k,1} = \dots  = T_{k,k-1} \geq \frac{1}{k}$.

We now show that $(d_n)$ must eventually have always the same sign,
i.e., there exists $n_0$ such that $\sign(d_m) = \sign(d_n)$ for 
all $m,n \geq n_0$. Note that by the assumption of non-zero partial sums,
we have $d_n \neq 0$ for all $n$.


\begin{lemma}
The $d_n$'s eventually have the same sign.
\end{lemma}

\begin{proof}
Let $c > 0$ be such that $\abs{c_n} < c$ for all $n$. Since $(d_n)$ is unbounded,
there exists $n_0$ such that $\abs{d_n} > c$ for all $n > n_0$ and then if there exists $m > n_0$ 
such that $d_m > 0$ and $d_{m+1} < 0$, we must have $d_m > c$ and $d_{m+1} < -c$.
Thus we have $ c_{m+1} = d_{m+1} - d_m < -2c$, and hence $\abs{c_{m+1}} > 2c$ 
which contradicts the boundedness assumption on $(c_n)$. 
\qed
\end{proof}
If the $d_n$'s are eventually negative then we use the sequence $\{c_n' = -c_n\}$ to obtain the same payoff 
and in this case $d_n = -\si c_i$ will be eventually positive. 
Therefore we assume that there is some $n_0$ such that $d_n > 0$ for all $n>n_0$. 
Let $\beta = max\{\abs{c_0},\abs{c_1},\dots ,\abs{c_{n_0}}\}$. We replace $c_0$ by 1 
and all $c_i$'s with $\beta$ for $1 \leq i \leq n_0$. By corollary \ref{cor:cchange} we observe that
the payoff function will still not change. Hence, we can also assume that $d_n > 0$ for all $n \geq 0$.

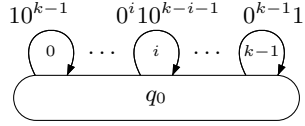
\begin{figure}[t]
\hrule
\begin{center}

\begin{picture}(45,17)(0,0)

\node[NLangle=0, Nw=38, Nh=6](q0)(22,2){$q_0$}

\gasset{loopdiam=6}

\drawloop[ELside=l,loopCW=y, loopangle=90](q0){$\quad 0^i10^{k-i-1}$}
\drawloop[ELside=r,loopCW=y, ELdist=2, loopangle=90](q0){${}_i$}

\node[Nframe=n, Nw=20, Nh=6](q0)(8,2){}

\drawloop[ELside=l,loopCW=y, loopangle=90](q0){$10^{k-1}\quad$}
\drawloop[ELside=r,loopCW=y, ELdist=2, loopangle=90](q0){${}_0$}

\node[Nframe=n, Nw=20, Nh=6](q0)(36,2){}

\drawloop[ELside=l,loopCW=y, loopangle=90](q0){$\quad 0^{k-1}1$}
\drawloop[ELside=r,loopCW=y, ELdist=2, loopangle=90](q0){${}_{k-1}$}

\node[Nframe=n](dots)(15,8){$\dots$}
\node[Nframe=n](dots)(29,8){$\dots$}

\end{picture}
\end{center}
\hrule
\caption{The game $G(k,i)$.}\label{game:gr_k}
\end{figure}

\begin{lemma}\label{lem:all-equal}
We have  $S_{k,i} = \frac{1}{k} = T_{k,i}$ for all $0 \leq i \leq k-1$.
\end{lemma}

\begin{proof}
Consider the game graph $G(k,i)$ which consists of state $q_0$
in which the player can choose among $k$ cycles of length $k$ 
where in the $i$th cycle, all rewards are $0$ except on the $(i+1)$th edge
which has reward $1$ (see \figurename~\ref{game:gr_k}).

Consider the strategy in state $q_0$ where the player after every $k \cdot r$ steps ($r \geq 0$) 
chooses the cycle which maximizes the contribution for the next $k$ edges. 
Let $i_r$ be the index such that $kr \leq i_r \leq kr + k-1$ and 
$c_{i_r} = \max\{c_{kr}, \dots, c_{kr+k-1}\}$ for $r \geq 0$. 
The payoff for this strategy is $\li t_n$ where 
$t_n = \frac{c_{i_0} + c_{i_1} + \dots + c_{i_{r-1}}}{d_{n}}$ for $i_{r-1} \leq n < i_r$. 

Note that $c_{i_r} \geq \frac{\sum_{i=kr}^{kr+k-1} c_i}{k}$ (the maximum is greater 
than the average), and we get the following (where $c$ is a bound on $(\abs{c_n})_{n \geq 0}$):
\[
\begin{array}{rcl} 
t_{n} & \geq & \displaystyle \frac{\sum_{i=0}^{n-1} c_i}{k \cdot d_{n}} - \frac{c}{d_{n}}, \\[2ex] 
\text{hence } \displaystyle \li t_{n} & \geq & \displaystyle \frac{1}{k} - \li \frac{c}{d_{n}} = \frac{1}{k}.
\end{array}
\]
By Lemma~\ref{lemm:sk0}, the payoff of all memoryless strategies in $G(k,i)$ 
is $S_{k,0}$, and the fact that memoryless optimal strategies exist entails that
$S_{k_0} = \li t_n \geq \frac{1}{k}$, and thus $S_{k,0} = \frac{1}{k} = S_{k,i}$ for all $0 \leq i \leq k-1$.

Using a similar argument on the graph $G(k,i)$ with reward $-1$ instead of $1$, 
we obtain $T_{k,0} = \frac{1}{k} = T_{k,i}$ for all $0 \leq i \leq k-1$.
\qed
\end{proof}

From Lemma~\ref{lem:all-equal}, it follows that
\[
S_{k,i} = \phi((0^i10^{k-i-1})^\omega) = 
\limn \frac{\sum_{r=0}^{\left[\frac{n}{k}\right]} c_{kr+i} }{d_n} = \frac{1}{k}
\] 
and hence,
\[
\begin{array}{rcl}
\phi\left((a_0a_1\dots a_{k-1})^\omega\right) & = & \displaystyle \li \sum_{i=0}^{k-1}  
\left( a_i \cdot \frac{\sum_{r=0}^{\left[\frac{n}{k}\right]} c_{kr+i} }{d_n} \right) 
=\sum_{i=0}^{k-1} \left( a_i \cdot \limn \frac{\sum_{r=0}^{\left[\frac{n}{k}\right]} c_{kr+i} }{d_n} \right) \\[2ex]
& = & \displaystyle \frac{\sum_{i=0}^{k-1} a_i }{k}.
\end{array}
\]

We show that the payoff of a regular word $u = b_1b_2\dots b_m (a_0a_1\dots a_{k-1})^\omega $ 
matches the mean-payoff value. 

\begin{lemma}\label{lemm:final}
If $u := b_1b_2 \dots b_m (a_0a_1 \dots a_{k-1})^\omega $ and $v = (a_0a_1 \dots a_{k-1})^\omega $ are two regular
sequences of weights then $\phi(u) = \phi(v) = \frac{\sum_{i=0}^{k-1} a_i }{k}$.
\end{lemma}

\begin{proof}
 Let $r \in \mathbb{N}$ be such that $kr > m$. If $\phi(v) \leq \phi(0v)$ then using Lemma \ref{prefc} we obtain 
$\phi(0v) \leq \phi(0^2v)$. Applying the lemma again and again, we get, 
$\phi(v) \leq \phi(0^m v) \leq \phi(0^{kr} v)$. From Corollary \ref{cor:preind} we obtain $\phi(0^mv) 
= \phi(b_1b_2 \dots b_m v) = \phi(u)$ and $ \phi(0^{kr} v) = \phi\left( (a_1a_2 \dots a_k)^{r} v\right) = \phi(v)$. Therefore, 
$\phi(u) = \phi(v) = \frac{\sum_{i=0}^{k-1} a_i }{k}$. The same argument goes through for the case 
$\phi(v) \geq \phi(0v)$.
\qed
\end{proof}

\begin{proof}[of Theorem~\ref{thm:mean}] 
In Lemma~\ref{lemm:final} we have shown that the payoff function $\phi$ must match 
the mean-payoff function for regular words, if the sequence $\{c_n\}_{n\geq 0}$
is bounded.
Since memoryless strategies in game graphs result in regular words over weights, 
it follows that the only payoff function that induces memoryless optimal
strategies is the mean-payoff function which concludes the proof.
\qed
\end{proof}

%

Observe that every regular sequence is bounded, and therefore the result of Theorem~\ref{thm:mean}
holds for all weighted average objectives with divergent sum defined by regular sequence
of weights.

\begin{corollary}
Let $(c_n)_{n \in \nat}$ be a regular sequence of real numbers with no zero partial sum such that 
$\sum_{i=0}^\infty \abs{c_i} = \infty$ (the sum is divergent). 
The weighted average payoff function $\phi$ defined by $(c_n)_{n \in \nat}$ 
induces optimal memoryless strategies for all two-player game graphs if and 
only if $\phi$ is the mean-payoff function.
\end{corollary}

\subsection{Unbounded sequence}
The results of Section~\ref{sec:converging} and Section~\ref{sec:bounded} can be summarized as follows:
(1) if the sum of $c_i$'s is convergent, then the sequence $\{\lambda^i\}_{i\geq 0}$, 
with $\lambda<1$ (discounted sum), is the only class of payoff functions that induce memoryless 
optimal strategies; and
(2) if the sum is divergent but the sequence $(c_n)$ is bounded, then the mean-payoff function 
is the only payoff function with memoryless optimal strategies (and the mean-payoff function 
is defined by the sequence $\{\lambda^i\}_{i\geq 0}$, with $\lambda=1$).
The remaining natural question is that if the sum is divergent and unbounded, 
then is the sequence $\{\lambda^i\}_{i\geq 0}$, with $\lambda>1$, the only
class that has memoryless optimal strategies. 
Below we show with an example that the class $\{\lambda^i\}$, with $\lambda>1$,
need not necessarily have memoryless optimal strategies.

We consider the payoff function given by the sequence $c_n = 2^n$. 
It is easy to verify that the sequence satisfies the partial non-zero assumption.
We show that the payoff function does not result into memoryless optimal 
strategies. To see this, we observe that the payoff for a regular word 
$w = b_0b_1 \dots b_t (a_0 a_1  \dots  a_{k-1})^\omega$ is given by $\min_{0 \leq i \leq k-1} \left( \frac{a_i + 2 a_{i+1} 
+  \dots  + 2^{k-1} a_{i+k-1}}{1+2+ \dots +2^{k-1}} \right) $ i.e., the payoff for a regular word is
the least possible weighted average payoff for its cycle considering all possible cyclic permutations
of its indices (note that the addition in indices is performed modulo $k$).

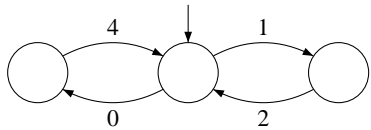
\begin{figure}
\hrule

\begin{center}
\begin{picture}(65,15)(0,0)

\node[Nmarks=i, iangle=90](n0)(35,5){}
\node(n1)(15,5){}
\node(n2)(55,5){}

\drawedge[curvedepth=4](n1,n0){4}
\drawedge[curvedepth=4](n0,n2){1}
\drawedge[curvedepth=4](n0,n1){0}
\drawedge[curvedepth=4](n2,n0){2}

\end{picture}
\end{center}

\hrule
\caption{The game $\mathcal{G}_{1024}$.} 
\label{game:1024}
\end{figure}

Now, consider the game graph $\mathcal{G}_{1024}$ shown in figure \ref{game:1024}. The payoffs for both the memoryless strategies
(choosing the left or the right edge in the start state) are $\min\left(\frac{5}{3},\frac{4}{3}\right)$ and 
$\min{\left(\frac{4}{3},\frac{8}{3}\right)}$
which are both equal to $\frac{4}{3}$. 
Although, if we consider the strategy which alternates between the two edges in the starting 
state then the payoff obtained is $\min{\left(\frac{37}{15},\frac{26}{15},\frac{28}{15},\frac{14}{15} \right)} = \frac{14}{15}$
 which is less than payoff for both the memoryless strategies. Hence, the player who minimizes the payoff
does not have a memoryless optimal strategy in the game $\mathcal{G}_{1024}$.
The example establishes that the sequence $\{2^n\}_{n\geq 0}$ does not induce optimal strategies.
 
\smallskip\noindent{\em Open question.} 
Though weighted average objectives such that the sequence is divergent and 
unbounded may not be of the greatest practical relevance, it is an interesting theoretical
question to characterize the subclass that induce memoryless strategies.
Our counter-example shows that $\{\lambda^n\}_{n \geq 0}$ with $\lambda>1$
is not in this subclass.

\end{document}